\documentclass[journal]{IEEEtran}
\usepackage{cite}

\ifCLASSINFOpdf
\else
\fi
\ifCLASSOPTIONcompsoc
 \usepackage[caption=false,font=normalsize,labelfont=sf,textfont=sf]{subfig}
\else
 \usepackage[caption=false,font=footnotesize]{subfig}
\fi

\ifCLASSOPTIONcaptionsoff
 \usepackage[nomarkers]{endfloat}
\let\MYoriglatexcaption\caption
\renewcommand{\caption}[2][\relax]{\MYoriglatexcaption[#2]{#2}}
\fi
\usepackage{amssymb,amsfonts}
\usepackage{amsmath}
\usepackage{cases}
\usepackage{subfig}
\usepackage{mathtools}

\DeclarePairedDelimiter\norm{\lVert}{\rVert}
\usepackage[makeroom]{cancel}
\makeatother
\usepackage{algorithmic}
\usepackage{graphicx,epstopdf}
\usepackage{tikz-cd}
\graphicspath{ {./images/} }

\DeclareGraphicsExtensions{.eps}

\usepackage{textcomp}
\usepackage{xcolor}
\usepackage{amsthm}
\def\BibTeX{{\rm B\kern-.05em{\sc i\kern-.025em b}\kern-.08em
    T\kern-.1667em\lower.7ex\hbox{E}\kern-.125emX}}
    
\newtheorem{theorem}{Theorem}

\newtheorem{assumption}{Assumption}

\newtheorem{remark}{Remark}

\begin{document}
%
\title{Composite Adaptive Control for Time-varying Systems with Dual Adaptation}
%
%
%

\author{Raghavv Goel
        and Sayan Basu Roy, Member, IEEE
\thanks{
Raghavv Goel is with the Robotics Institute, Carnegie Mellon University, USA, 15213; e-mail: raghavvg@andrew.cmu.edu
\\
Sayan Basu Roy is with the Department
of Electronics and Communication Engineering, Indraprastha Institute of Technology Delhi, 110020, India  e-mail: sayan@iiitd.ac.in}
}

\maketitle

\begin{abstract}
This paper proposes a composite adaptive control architecture using dual adaptation scheme for dynamical systems comprising time-varying uncertain parameters. While majority of the adaptive control schemes in literature address the case of constant parameters, recent research has conceptualized improved adaptive control techniques for time-varying systems  with rigorous stability proofs. The proposed work is an effort towards a similar direction, where a novel dual adaptation mechanism is introduced to efficiently tackle the time-varying nature of the parameters. Projection and $\sigma$-modification algorithms are strategically combined using congelation of variables to claim a global result for the tracking error space. While the classical adaptive systems demand a restrictive condition of persistence of excitation (PE) for accurate parameter estimation, the proposed work relies on a milder condition, called initial excitation (IE) for the same. A rigorous Lyapunov stability analysis is carried out to establish uniformly ultimately bounded (UUB) stability of the closed-loop system. Further it is analytically shown that the proposed work can recover the performance of previously designed IE-based adaptive controller in case of time invariant systems.
\end{abstract}


%
\IEEEpeerreviewmaketitle

\section{Introduction}
%
%
%
%
\IEEEPARstart{A}{daptive} control is a powerful nonlinear dynamic control technique, which can tackle parametric uncertainty in real-time \cite{narendra2012stable, ioannou1996robust}.  Adaptive controllers ensure closed-loop stability of the extended error dynamics involving tracking error and parameter estimation error. While asymptotic tracking error convergence can be claimed by invoking Barbalat's lemma, parameter convergence demands an additional restrictive condition, called persistence of excitation (PE), on the regressor signal.

Majority of the developments in adaptive control literature consider constant unknown parameters to establish well-behaved closed-loop error dynamics. Compared to the mammoth parameter estimation literature for constant parameters, the time-varying parameter estimation literature shies away. The stability analysis for the case of unknown time-varying parameters is challenging due to the appearance of an undesirable parameter derivative term in the Lyapunov analysis (for details see the introduction of \cite{Dixon}). One approach for bounding the undesirable term is to use robust adaptive control techniques like sigma-modification \cite{lavretsky2013robust}, projection \cite{lavretsky2011projection} etc, especially for slow-varying parameters. However, recent literature has provided promising results in adaptive control for time-varying systems \cite{chen2020adaptive2, Dixon}, where robust damping and/or sliding-mode like mechanisms are strategically utilized to ensure improved performance. Further the work in \cite{arabi2019set} uses barrier Lyapunov function to invoke safety bounds on the tracking error in the context of time-varying parameters.\par 

Most of the above mentioned frameworks/techniques have only proved tracking  error convergence, while  parameter estimation error convergence requires the PE condition on the state/reference input \cite{narendra2012stable}. It has been well-established in literature that the PE condition has difficulty to verify and/or satisfy in practical problems. Composite adaptive control techniques \cite{slotine1989composite, lavretsky_composite} provide a way to improve the parameter estimation algorithm by incorporating prediction error (partial information about parameter estimation error) in addition to tracking error. However, these techniques still require the PE condition for parameter convergence. Research efforts are made in recent past to relax the PE condition in various ways, such as data-driven \cite{chowdhary2013concurrent, kamalapurkar2017concurrent, parikh2019integral}, filter-based \cite{ortega2020new, roy2017combined} methods. The works in \cite{roy2016parameter, roy2017combined, roy2017uges, roy2019robustness} have devised a condition, called Initial Excitation (IE), which is shown to be sufficient for parameter convergence using a two-tier filter-based adaptive controller. The IE condition is milder than the classical PE condition since it requires the excitation to sustain only in initial time-window as compared to PE demanding the excitation to sustain for all time. 

Unlike the above mentioned IE-based algorithms, which are proved to be efficient for systems with constant unknown parameters, the work in \cite{gaudio2021parameter} has devised a novel adaptive controller for time-varying systems while ensuring parameter convergence (to an ultimate-bound) under IE condition. However, this algorithm cannot ensure restoring the asymptotic tracking performance in the case of constant parameters. 


Taking inspiration from the recent works \cite{chen2018adaptive, chen2020adaptive2}, this paper utilises the concept of \textit{congelation of variables}, which splits the unknown parameter vector/matrix into a nominal component (constant) and a perturbation component (time-varying). A composite adaptive control architecture for time-varying systems is developed using a dual adaptation scheme. The dual adaptation mechanism comprises 1) a primary estimator for estimating the total unknown parameter vector/matrix  and 2) a secondary estimator dedicated for the constant nominal component. Projection and $\sigma$-modification algorithms are strategically combined in the adaptation mechanism to claim a global result for the tracking error space. The secondary parameter estimator utilises the notion of IE condition for efficient learning of the nominal component of the parameter. A rigorous Lyapunov stability analysis is carried out to establish uniformly ultimately bounded (UUB) stability of the closed-loop system. It is analytically proved that the proposed algorithm can recover the performance of previously designed IE-based adaptive controllers \cite{roy2017combined, roy2017uges} in case of plants involving constant parameters only. In a nutshell, the paper has the following contributions. \par 
\begin{itemize}
    \item A novel dual adaptation mechanism utilising congelation of variables for uncertain time-varying systems.
    \item Strategically combining $\sigma$-mod and projection to ensure global tracking (in the tracking error space), unlike a semi-global result in \cite{pan2018composite}.
    \item Extension of the IE-based adaptation strategy for time-varying system, while attaining performance recovery for the time-invariant case, unlike \cite{gaudio2021parameter}. 
\end{itemize}




\section{Preliminaries}
\label{sec:prelim}
Some of the notations and definitions used throughout the paper are stated. For a vector $a$, $\norm{a}$ denotes the Euclidean norm. For a $n\times n$ matrix $A$,   $\norm{A}_{F}$ denotes the Frobenius norm and $Tr(A)$ denotes the trace of $A$. $\mathbf{1}_{n},\mathbf{0}_{n} \in \mathbb{R}^{n}$ are column vectors having all entries as $1$ and $0$, respectively; $I_{n}$ is the identity matrix of dimension $n\times n$.  $\mathbf{0}_{n \times m}, \mathbf{1}_{n \times m} \in \mathbb{R}^{n\times m}$ are matrices with all entries as $0$ and $1$ respectively. $\mathbb{S}_{n}^{+}$ is the set of symmetric positive-definite matrices of size $n\times n$.


\section{System Description}
\label{sec:sytem}
Consider the following dynamical system \cite{Yucelen_acc}
\begin{align}
    \dot{x}(t) = Ax(t) + B\big(u(t) + W^{T}(t)\phi(x(t))\big)
    \label{eq:x_dot}
\end{align}
where, $A \in \mathbb{R}^{n\times n}, B \in \mathbb{R}^{n\times n_{u}}$ are the system matrices, $x(t) \in \mathbb{R}^{n}$ is the system state, $u(t) \in \mathbb{R}^{n_{u}}$ is the control input, $W(t) \in \mathbb{R}^{n_{w} \times n_{u}}$ is the unknown time-varying parameter and $\phi(x(t)) \in \mathbb{R}^{n_{w}}$ is a known regressor and function of the state. 

To characterize the desired response a reference model is designed below.
\begin{align}
    \dot{x}_{m}(t) = A_{m}x_{m}(t) + B_{m}r(t)
    \label{eq:reference_model}
\end{align}
where, $x_{m}(t) \in \mathbb{R}^{n}$ is the reference model state, $A_{m} \in \mathbb{R}^{n \times n}, B_{m} \in \mathbb{R}^{n \times n_{r}}$ are reference model matrices. The matrix $A_m$ is designed to be Hurwitz to ensure bounded-input-bounded-output (BIBO) stability of (\ref{eq:reference_model}) with respect to the piecewise-continuous external reference input $r(t) \in \mathbb{R}^{n_{r}}$.

\section{Control Objective and Assumptions}
\label{sec:main}
The objective is to design a control law $u(t)$ and parameter update law $\dot{\hat{W}}(t)$ such that the closed-loop error dynamics including tracking error $e(t)\triangleq x(t)-x_m(t)$ and parameter estimation error $\tilde{W}(t)\triangleq\hat{W}(t)-W(t)$ remain uniformly ultimately bounded (UUB).

The following assumptions are made to facilitate the design.
\begin{assumption}
System matrix $A$ and $B$ have the following matching conditions: $A=A_{m}-BK_{x}^{T}$ and $B_{m} = BK_{r}^{T}$, where, $K_{x} \in \mathbb{R}^{n \times n_{u}}$ and $K_{r} \in \mathbb{R}^{n_{r} \times n_{u}}$ are called controller parameters.
 \label{ass:matching_condition}
\end{assumption}
The above is a standard assumption in MRAC literature \cite{narendra2012stable}, which ensures structural similarity between the plant and the reference model. 
Further based on congelation of variables method in \cite{chen2018adaptive}, the time-varying parameter $W(t)$ can be decomposed as
\begin{align}
    W(t) = W^{*} + \delta_{W}(t)
    \label{eq:Wt_congelation}
\end{align}
where $W^*$ is the constant nominal component of the parameter and $\delta_{W}(t)$ is the perturbation component (deviation around the nominal).
\begin{assumption}
The time-varying parameter can be norm-bounded by known constants as: $\norm{W^{*}} \leq \overline{W}$, $\norm{\delta_{W}(t)} \leq \Bar{\delta}_{W}$ and $\norm{\dot{\delta}_{W}(t)} \leq \Bar{\dot{\delta}}_{W}$.
\label{ass:W_bound}
\end{assumption}


\section{Adaptive Controller Design}
\subsection{Control Input}
The control input is designed as follows.
\begin{align}
    u(t) = K_{x}^{T}x(t) + K_{r}^{T}r(t) - \underbrace{\hat{W}^{T}(t)\phi(x(t))}_{u_{ad}(t)}
    \label{eq:controller}
\end{align}
where $u_{ad}(t)$ is the adaptive component of the controller. We assume known $K_{x}$ and $K_{r}$, as estimation and convergence of these parameters are not the focus of current paper and can be handled along the lines of \cite{roy2017combined}. We thus emphasise only on the estimation of the unknown time-varying parameter ($W(t)$).

Using \eqref{eq:x_dot}, \eqref{eq:reference_model} and \eqref{eq:controller}, the closed-loop error dynamics is given as
\begin{align}
\begin{split}
    \dot{e}(t) =& A_{m}e(t) - B\Tilde{W}^{T}(t)\phi(x(t)) 
\end{split}
    \label{eq:e_dot_closed_loop}
\end{align}

The parameter update law is subsequently proposed using a dual adaptation mechanism - 
\begin{itemize}
    \item a Primary Estimate ($\hat{W}(t)\in\mathbb{R}^{n_w\times n_u}$) for the total time varying parameter ($W(t)$)
    \item a Secondary Estimate ($\hat{W}^{*}(t)\in\mathbb{R}^{n_w\times n_u}$) for the nominal component of the parameter ($W^{*}$)
\end{itemize} 
\subsection{Primary Parameter Estimator $\hat{W}(t)$}
The parameter estimate for $W(t)$ is designed using gamma-Projection operation (Definition (11) of \cite{lavretsky2011projection}) which is  denoted as $\text{Proj}_{\Gamma_{W}}(\hat{W}(t),y(t),f(t))$, where
\begin{align}
    f(t) =& \frac{Tr(\hat{W}^{T}(t)\hat{W}(t)) - \alpha^{2}}{2\alpha\epsilon + \epsilon^{2}} \label{eq:w_hat_convex_fn}
    \\
    y(t) =& \phi(x(t))e^{T}(t)PB - \sigma(\hat{W}(t)-\hat{W}^{*}(t))
    \label{eq:y}
\end{align}
Here, $f(t) \in \mathbb{R}$ is a convex function,  $\alpha^{2} \triangleq \overline{W}^{2} + \Bar{\delta}_{W}^{2}$ is from Assumption \ref{ass:W_bound} and $\epsilon \in \mathbb{R}_{>0}$ is a continuity parameter ensuring Lipchitz continuity. $\Gamma_{W} \in \mathbb{S}_{n_{w}}^{+}$ is the adaptation gain matrix and $\sigma \in \mathbb{R}_{>0}$ is a scalar tuner similar to sigma-mod \cite{lavretsky2013robust}; $y(t) \in \mathbb{R}^{n_{w} \times n_{u}}$ and $P\in\mathbb{R}^{n\times n}$ is a positive definite solution of the Lyapunov equation
\begin{equation}
    A_m^TP+PA_m=-Q_m
\end{equation}
where $Q_m>0$ is a chosen positive definite matrix. Further $\hat{W}^*(t)$ is the nominal parameter estimate designed subsequently. Hence, 
\begin{subnumcases}{\dot{\hat{W}}(t)=}
  \Gamma y(t) - \Gamma\frac{(\nabla f(t))(\nabla f(t))^{T}}{Tr\big((\nabla f(t))^{T}\Gamma(\nabla f(t))\big)}\Gamma y(t)f(t) \label{eq:proj1}
  \\
  \Gamma y(t)\label{eq:proj2}
\end{subnumcases}  
\begin{align}
 =&\text{Proj}_{\Gamma_{W}}(\hat{W}(t),y(t),f(t))
 \label{eq:W_hat_dot}
\end{align}
where, (\ref{eq:proj1}) occurs if $f(t)>0 \wedge Tr\big(y^{T}(t)\Gamma(\nabla f(t)\big)>0$, otherwise (\ref{eq:proj2}) occurs.

\begin{remark}
The design in \eqref{eq:y} is a novel concept, where the second term includes $\sigma$-modification ($\sigma$ mod) while pulling the primary estimate $\hat{W}(t)$ towards the secondary estimate $\hat{W}^*(t)$. Hence, this update law can obviate the drawback of unlearning that exists in traditional $\sigma$-mod, provided that the secondary estimate approaches the nominal parameter $W^*$. \end{remark}

\subsection{Secondary Parameter Estimator $\hat{W}^{*}(t)$}
The design of update law for $\hat{W}^{*}(t)$ is another novel contribution of the paper. The design is motivated from the recent literature on IE-based adaptive control \cite{roy2017combined, roy2017uges}, which builds on two-tier filter architecture while ensuring parameter convergence without the restrictive PE condition on the regressor.  The two-tier filtering scheme is subsequently adopted for the secondary estimate while suitably modifying the scheme in the context of time-varying parameter setting. 


Define the secondary parameter estimator as 
\begin{align}
    \Tilde{W}^{*}(t) = \hat{W}^{*}(t) - W^{*}
    \label{eq:lw_tilde}
\end{align}

\subsubsection{First-layer Filtering}
Exploiting the idea of congelation of variables \cite{chen2018adaptive}, we define a series of filters to extract information about the nominal component of the unknown parameter.
\begin{align}
    &\dot{g}(t)= -p_{f}g(t) + \dot{e}(t),\ g(t_{0}) = \mathbf{0}_{n} 
    \label{eq:g_dot}
    \\
    &\dot{e}_{f}(t) = -p_{f}e_{f}(t) + e(t),\  e_{f}(t_{0}) = \mathbf{0}_{n}
    \label{eq:ef_dot}
    \\
    &\dot{u}_{f}(t) = -p_{f} u_{f}(t) + u_{ad}(t),\  u_{f}(t_{0}) = \mathbf{0}_{n_{u}}
    \label{eq:uf_dot}
    \\
    &\dot{\phi}_{f}(t) = -p_{f}\phi_{f}(t) + \phi(x(t)),\  \phi_{f}(x(t_{0})) = \mathbf{0}_{n_{w}} 
    \label{eq:phi_f_dot}
\end{align}
where, $p_{f} \in \mathbb{R}_{>0}$ determines the weight given to past trajectories. $g(t) \in \mathbb{R}^{n}$ is the filtered-tracking error derivative, $e_{f}(t) \in \mathbb{R}^{n}, u_{f}(t) \in \mathbb{R}^{n_{u}}$ and $\phi_{f}(t) \in \mathbb{R}^{n_{w}}$ are the filtered-tracking error, filtered-control input and filtered-regressor respectively.
Note that the time derivative of trajectory error ($\dot{e}(t)$) is not available. Therefore, 
$g(t)$ is calculated using Integration By-parts where no information of $\dot{e}(t)$ is required and is given as
\begin{align}
    g(t) = e(t) - exp\{-p_{f}(t-t_{0})\}e(t_{0}) - p_{f}e_{f}(t)
    \label{eq:g}
\end{align}

The equation \eqref{eq:e_dot_closed_loop} is re-written as
\begin{align}
\begin{split}
    W^{*T}\phi(x(t)) =& \underbrace{(B^{T}B)^{-1}B}_{\Bar{B}}\big(\dot{e}(t) - A_{m}e(t)\big)
    \\
    & - \delta_{W}^{T}(t)\phi(x(t)) + \underbrace{\hat{W}^{T}(t)\phi(x(t))}_{u_{ad}}
\end{split}
\label{eq:e_dot_filter1}
\end{align}

Substituting \eqref{eq:ef_dot}-\eqref{eq:g} in \eqref{eq:e_dot_filter1}, we get
\begin{align}
    W^{*T}\phi_{f}(t) + \Delta_{f}(t) = \underbrace{\Bar{B}(g(t)-A_{m}e_{f}(t))}_{h(t)} + u_{f}(t) 
    \label{eq:first_layer_filtered_W_star}
\end{align}
where, the term corresponding to time-varying component of the parameter satisfies the following dynamics.
\begin{align}
    \dot{\Delta}_{f}(t) &= -p_{f}\Delta_{f}(t) + \delta_{W}^{T}(t)\phi(x(t)) ,\, \Delta_{f}(t_{0}) = \mathbf{0}_{n_{u}}      
    \label{eq:Delta_dot}    
\end{align}
Here $\Delta_{f}(t) \in \mathbb{R}^{n_{u}}$ is the filtered effect of the parametric perturbation $\delta_W(t)$.

It can be observed that the first-layer filtering provides an algebraic relation \eqref{eq:first_layer_filtered_W_star} involving $W^*$ as compared to the differential relation \eqref{eq:e_dot_closed_loop} having unmeasurable quantity $\dot{e}(t)$. Hence, relation \eqref{eq:first_layer_filtered_W_star} can be utilised to design a composite adaptive controller. However, we further define another layer of filter to exploit the benefit of the IE condition \cite{roy2016parameter}.
\subsubsection{Second-layer Filtering}
Consider the following filter-dynamics, which take outer-product of first-layer filter outputs as inputs.  
\begin{align}
\begin{split}
        \dot{\Phi}_{ff}(t) =& -p_{ff}\Phi_{ff}(t) + \phi_{f}(t)\phi_{f}^{T}(t),
        \\
        & \Phi_{ff}(t_{0}) = \mathbf{0}_{n \times n}
\end{split}
    \label{eq:M_dot}
    \\
    \begin{split}
    \dot{u}_{ff}(t) =& -p_{ff}u_{ff}(t) + (h(t)+u_{f}(t))\phi_{f}^{T}(t), 
    \\
    & u_{ff}(t_{0}) = \mathbf{0}_{n_{u} \times n} 
    \end{split} \label{eq:G_dot}
\end{align}
where $\Phi_{ff}(t)\in\mathbb{R}^{n_{w} \times n_{w}}$ and $u_{ff}(t) \in \mathbb{R}^{n_{u} \times n_{w}}$ are the double filtered regressor and double filtered control input respectively.

Using (\ref{eq:first_layer_filtered_W_star}), \eqref{eq:M_dot} and \eqref{eq:G_dot}, it can be shown that 
\begin{align}
    u_{ff}(t) = & W^{*T}\Phi_{ff}(t) + \Delta_{ff}(t)
    \label{eq:second_layer_filter_relation}
\end{align}
where,
\begin{align}
\begin{split}
    \dot{\Delta}_{ff}(t) = & -p_{ff}\Delta_{ff}(t) + \Delta_{f}(t)\phi_{f}^{T}(t),
    \\
    & \Delta_{ff}(t_{0})=\mathbf{0}_{n_{u} \times n_{w}}
\end{split}
    \label{eq:delta_ff}
\end{align}
Here, $\Delta_{ff}(t) \in \mathbb{R}^{n_{u} \times n_{w}}$ is the double-filtered effect of the parametric perturbation $\delta_{w}(t)$.

\subsubsection{Initial Excitation}
Consider the following IE assumption on the filtered-regressor. 
\begin{assumption} \label{ass:IE}
The filtered-regressor $\phi_{f}(x(t))$ is uniformly initially exciting (u-IE) with respect to dynamics in \eqref{eq:x_dot}, filters in \eqref{eq:ef_dot}, \eqref{eq:uf_dot}, \eqref{eq:phi_f_dot}, \eqref{eq:M_dot}, \eqref{eq:G_dot} and \eqref{eq:y_star_CIE} with time-window $T_{IE}$ and degree of excitation $\gamma_{IE}$, i.e., $\exists\ \gamma_{IE}>0, T_{IE}>0$ such that 
\begin{align}
    \int_{t_0}^{t_0+T_{IE}}\phi_f(\tau)\phi_f^T(\tau)d\tau\geq \gamma_{IE}I_{n_w}
\end{align}
where $I_{n_w}$ is the identity matrix of dimension $n_{w}$.
\end{assumption}
\begin{remark}
The definition of IE condition \cite{roy2017uges,roy2019robustness} has a crucial difference with the definition of PE condition. In PE condition, a similar integral inequality has to be satisfied for $[t,t+T_{PE}]$, $\forall t\in[t_0,\infty)$, i.e., the excitation has to persist for all future time. Unlike PE, the IE condition demands the integral inequality only for the initial time-window $[t_0,t_0+T_{IE}]$. The IE condition is milder than PE since there is no need for the excitation to persist beyond initial time-window. 
\end{remark}
Dynamics of $\hat{W}^{*}(t)$ is designed using the gamma-Projection operation similar to \eqref{eq:W_hat_dot}, while incorporating an IE-based component. 
\begin{align}
    \dot{\hat{W}}^{*}(t) = \text{Proj}_{\Gamma_{W^{*}}}(\hat{W}^{*}(t),y^{*}(t),f^{*}(t))
    \label{eq:W_hat_star_dot}
\end{align}
where, 
\begin{align}
        y^{*}(t) =& \gamma_{1}C_{l}(t) + \gamma_{2}C_{ll}(t) + \gamma_{3}s(t)C_{IE}(t)
        \label{eq:y_star}
        \\
        f^{*}(t) =& \frac{Tr(\hat{W}^{*T}(t)\hat{W}^{*}(t))-\alpha^{*2}}{{2\alpha^{*}\epsilon^{*} + \epsilon^{*2}}}
\end{align}
where, $\alpha^{*} \triangleq \overline{W}$ is the known upper bound of $W^{*}$ from Assumption \ref{ass:W_bound} and $\epsilon^{*} \in \mathbb{R}_{>0}$ is a continuity parameter ensuring Lipchitz condition. $\Gamma_{W^{*}} \in \mathbb{S}_{n_{w}}^{+}$ is the adaptation gain matrix and $\gamma_{1}, \gamma_{2}, \gamma_{3} \in \mathbb{R}_{>0}$ are parameter tuning scalars for individual filter terms.

where,    
\begin{align}
    &C_{l}(t) = -\phi_{f}(x(t))\big(\hat{W}^{*T}(t)\phi_{f}(x(t)) - (h + u_{f}) \big)^{T}\\
    &C_{ll}(t) = -\big(\hat{W}^{*T}(t)\Phi_{ff}(t) - u_{ff}(t)\big)^{T}
    \\
    &C_{IE}(t) = - \big(\hat{W}^{*T}(t)\Phi_{ff}(T) - u_{ff}(T)\big)^{T}
\end{align}

where, $T\triangleq t_0+T_{IE}$; the switching signal $s(t) = 0$ if $t \in [t_{0}, t_{0} + T_{IE})$, otherwise $1$ and 
$y^{*}(t) \in \mathbb{R}^{n_{w}\times n_{u}}$.  
\begin{remark}
It is proved in \cite{roy2017uges,roy2019robustness} that the IE condition can be verified online by checking the minimum eigen-value of $\phi_{ff}(t)$. Hence, the above designed parameter estimator is online implementable. 
\end{remark}

\section{Stability Analysis}
\subsection{Ultimate Boundedness of $e(t)$ and $\tilde{W}(t)$}
\begin{theorem}
Using the system model in \eqref{eq:x_dot}, control design in \eqref{eq:controller}, parameter update laws in \eqref{eq:W_hat_dot} and \eqref{eq:W_hat_star_dot} and Assumptions \ref{ass:matching_condition}-\ref{ass:W_bound}. The system is uniformly ultimately bounded (UUB) in the extended state space of $[e^{T}(t) \, \Tilde{W}^{T}(t)]$.
\label{theorem:1}
\end{theorem}
\begin{proof}
The Lyapunov candidate is defined as
\begin{align}
    V(e(t),\Tilde{W}(t)) = e^{T}(t)Pe(t) + Tr\big(\Tilde{W}^{T}(t)\Gamma_{W}^{-1}\Tilde{W}(t)\big)
    \label{eq:V_candidate}
\end{align}

Taking derivative along system trajectories and specifying explicit time dependence wherever necessary, we get
\begin{align}
    \dot{V} =& e^{T}(t)(A_{m}^{T}P + PA_{m})e(t) - 2e^{T}(t)PB\Tilde{W}^{T}(t)\phi(x(t)) \nonumber    
    \\
    & + 2Tr\Big(\Tilde{W}^{T}(t)\big(\phi(x(t))e^{T}(t)PB - \sigma(\hat{W}(t)-\hat{W}^{*}(t))\big) \nonumber
    \\
    &  \underbrace{-\Tilde{W}^{T}(t)\big(\Gamma_{W}^{-1}\text{Proj}_{\Gamma_{W}}(\Tilde{W},y,f) - y(t)\big)}_{\leq 0 \text{\cite{lavretsky2011projection}}} \Big) \nonumber
    \\
    & - 2Tr\big(\Tilde{W}^{T}(t)\Gamma_{W}^{-1}\dot{W}(t)\big) \nonumber
\end{align}
After cancelling like terms, bounding Projection term and using $\dot{W}(t) = \dot{\delta}_{W}(t)$ from \eqref{eq:Wt_congelation}, we get
\begin{align}
    \dot{V}\leq & e^{T}(t)(A_{m}^{T}P + PA_{m})e(t) - 2\sigma Tr\big(\Tilde{W}^{T}(t) (\hat{W}(t)- \nonumber
    \\
    & \hat{W}^{*}(t))\big) - 2Tr\big(\Tilde{W}^{T}(t)\Gamma_{W}^{-1}\dot{\delta}_{W}(t)\big) \nonumber
    \\
    \begin{split}
    = & -e^{T}(t)Q_{m}e(t) - 2\sigma Tr\big(\Tilde{W}^{T}(t)\Tilde{W}(t) +  
    \\
    & \Tilde{W}^{T}(t)(W(t) - \hat{W}^{*}(t)-\Gamma_{W}^{-1}\dot{\delta}_{W}(t))\big)         
    \end{split} \label{eq:lw_tilde_bound}
    \\
    \leq & -\lambda_{\min}(Q_{m})\norm{e(t)}_{2}^{2} - 2\sigma Tr\big(\Tilde{W}^{T}(t)\Tilde{W}(t)\big) +  \nonumber
    \\
    & \sigma\norm{\Tilde{W}}_{F}^{2} + \sigma\norm{W_{r}(t)}_{F}^{2} \nonumber
    \\
    \leq & -\lambda_{\min}(Q_{m})\norm{e(t)}_{2}^{2} - \sigma \norm{\Tilde{W}(t)}_{F}^{2} +  \sigma\norm{W_{r}(t)}_{F}^{2} \nonumber 
    \\
    \leq & -\underbrace{\min(\lambda_{\min}(Q_{m}), \sigma)}_{\beta_{1}}\big(\norm{e(t)}^{2}+ \norm{\Tilde{W}(t)}_{F}^{2}\big) + c_{W} \label{eq:V_dot_part1}
\end{align}

where, $W_{r}(t)=W(t)-\hat{W}^{*}(t) -\Gamma_{W}^{-1}\dot{\delta}_{w}(t)$,  here all individual terms in $W_{r}(t)$ can be upper-bounded from Assumption \ref{ass:W_bound} and use of projection operator implying that $c_{W}$ is  finite positive constant: $\sigma \norm{W_{r}(t)}_{F}^{2} \leq \sigma(2\overline{W} + \bar{\delta}_{W} + \lambda_{\min}(\Gamma_{W})\Bar{\dot{\delta}}_{W} + \epsilon^{*}) = c_{W}$; further $\lambda_{\min}(Q_m)>0$ is the minimum eigenvalue of the matrix. 

From \eqref{eq:V_candidate}, the LHS can be bounded as 
\begin{align}
\begin{split}
        V \leq \beta_{2}\big(\norm{e(t)}^{2} + \norm{\Tilde{W}(t)}_{F}^{2}\big)
        \\
        \implies -\big(\norm{e(t)}^{2} + \norm{\Tilde{W}(t)}_{F}^{2}\big) \leq -\frac{1}{\beta_{2}}V
\end{split}
\label{eq:Lyap_bound}
\end{align}

where, $\beta_{2} = max\big(\lambda_{max}(P), \lambda_{min}(\Gamma_{W})\big)$ 

Using \eqref{eq:Lyap_bound} in \eqref{eq:V_dot_part1}

\begin{align}
    \dot{V} \leq -\frac{\beta_{1}}{\beta_{2}}V + c_{W}
    \label{eq:V_dot_UUB}
\end{align}



From \textit{theorem 4.18} in \cite{khalil2002nonlinear}, the solution of the combined error system $[e^{T}(t) \Tilde{W}^{T}(t)]$ is UUB. 

\end{proof}

\begin{remark}
Th term involving $\hat{W}^{*}(t)$ in the above Lyapunov analysis was possible to tackle only because of the Proj$_{\Gamma}$(.) operator, i.e., $\norm{\hat{W}^{*}(t)}\leq \alpha^*+\epsilon^*$ $\forall t\in[t_0,\infty)$. The above analysis ensures a UUB result in a global sense in the tracking error space even though the nominal estimator dynamics $\dot{\hat{W}}^*(t)$ is perturbed by state-dependent disturbances ($\Delta_f(t)$ and $\Delta_{ff}(t)$). This is in contrast to \cite{pan2018composite}, which claims a semi-global result due to the presence of similar state-dependent disturbances.  Furthermore, the bound $c_{W}$ can be reduced by designing precise estimate of the nominal parameter, i.e., $\Tilde{W}^{*}(t)\approx 0$ implies $c_{W} \approx \sigma\norm{\delta_{w}(t)-\Gamma_{W}^{-1}\dot{\delta}_{w}(t)}_{F}^{2}$. Hence, the ultimate-bound would become only dependent on the disturbance bounds, unlike traditional $\sigma$-mod, where the ultimate-bound is also dependent on the parameter upper-bound $\overline{W}$.  

\end{remark}


\subsection{Ultimate Boundedness of $\tilde{W}^*(t)$}
We re-write \eqref{eq:y_star} as the following $\forall t\in[t_0+T_{IE},\infty)$. 
\begin{align}
    y^{*}(t) = \Gamma_{W^{*}}\big(\gamma_{1}C_{l}(t) + \gamma_{2}C_{ll}(t)+ \gamma_{3}C_{IE}(t)\big) 
    \label{eq:y_star_CIE}
\end{align}

\begin{theorem}

Based on the nominal parameter estimation update law in \eqref{eq:W_hat_star_dot} in the presence of Assumption \ref{ass:IE}, $\tilde{W}^{*}(t)$ is uniformly ultimately bounded (UUB).
\end{theorem}
\begin{proof}
The Lyapunov candidate is defined as 
\begin{align}
    V^{*} = \frac{1}{2}Tr(\tilde{W}^{*T}(t)\Gamma_{W^{*}}^{-1}\tilde{W}^{*}(t))
    \label{eq:V_star}
\end{align}

Taking time-derivative along system trajectories $\forall t\in[t_0+T_{IE},\infty)$ yields 
\begin{align}
    \dot{V}^{*} =& -\gamma_{1}Tr\big(\tilde{W}^{*T}(t)(\phi_{f}(t)\phi_{f}^{T}(t)\tilde{W}^{*}(t) - \phi_{f}(t)\Delta_{f}^{T}(t))\big) \nonumber
    \\
    & -\gamma_{2}Tr\big(\tilde{W}^{*T}(t)(\Tilde{W}^{*T}(t)\Phi_{ff}(t) -  \Delta_{ff}(t))^{T}\big) \nonumber    
    \\
    & -\gamma_{3}Tr\big(\tilde{W}^{*T}(t)(\Tilde{W}^{*T}(t)\Phi_{ff}(T) -  \Delta_{ff}(T))^{T}\big) \nonumber
\end{align}
Using \eqref{eq:phi_f_dot}, \eqref{eq:Delta_dot}, \eqref{eq:delta_ff} and bounding $\norm{\phi(x(t))} \leq \bar{\phi}$ as $x(t) \in \mathcal{L}_{\infty}$ from Theorem \ref{theorem:1}, we get the following bounds on the undesirable terms: $\norm{\phi_{f}(t)} \leq \frac{\bar{\phi}}{p_{f}}, \, \norm{\Delta_{f}(t)} \leq \frac{\bar{\delta}_{W}\bar{\phi}}{p_{f}}, \, \norm{\Delta_{ff}(t)} \leq \frac{\bar{\delta}_{W}\bar{\phi}^{2}}{p_{f}p_{ff}}$, which gives
\begin{align}
    \dot{V}^{*} \leq & -Tr\big(\tilde{W}^{*T}(t)(\gamma_{2}\Phi_{ff}(t) + \gamma_{3}\Phi_{ff}(T))\tilde{W}^{*}(t)\big) \nonumber
    \\
    & +\norm{\tilde{W}^{*}(t)}_{F}\underbrace{\big(\gamma_{1}\frac{\bar{\delta}_{W}\bar{\phi}^{2}}{p_{f}p_{f}} + (\gamma_{2}+\gamma_{3})\frac{\bar{\delta}_{W}\bar{\phi}^{2}}{p_{f}p_{ff}}\big)}_{c^{*}} \label{eq:c_star}
    \\
     & \text{ as $\Phi_{ff}(T)$ is full rank from Assumption \ref{ass:IE}} \nonumber
    \\
    \leq & - \underbrace{\gamma_{3}\lambda_{\min}(\Phi_{ff}(T))}_{\beta_{1}^{*}}\norm{\tilde{W}^{*}(t)}^{2}_{F} + c^{*}\norm{\tilde{W}^{*}(t)}_{F} \label{eq:bound_on_Wtilde_star}
    \\
    = & -\frac{\beta_{1}^{*}}{2}\norm{\tilde{W}^{*}(t)}^{2}_{F} -\big(\sqrt{\frac{\beta_{1}^{*}}{2}}\norm{\tilde{W}^{*}(t)}_{F} - \frac{c^{*}}{\sqrt{2\beta_{1}^{*}}}\big)^{2} \nonumber
    \\
    & + \frac{c^{*2}}{2\beta_{1}^{*}} \nonumber
    \\
    \leq & -\frac{\beta_{1}^{*}}{2}\norm{\tilde{W}^{*}(t)}^{2}_{F} + \frac{c^{*2}}{2\beta_{1}^{*}} 
    \label{eq:lyap_star_dot}
\end{align}

From \eqref{eq:V_star}, we get
\begin{align}
    V^{*} \leq \frac{\lambda_{\max}(\Gamma_{W^{*}}^{-1})}{2}\norm{\tilde{W}^{*}(t)}^{2}_{F}
    \label{eq:V_star_bound}
\end{align}

Substituting \eqref{eq:V_star_bound} in \eqref{eq:lyap_star_dot} and putting $\beta_{2}^{*} = \lambda_{\max}(\Gamma_{W^{*}}^{-1})$, we finally get
\begin{align}
    \dot{V}^{*} \leq -\frac{\beta_{1}^{*}}{\beta_{2}^{*}}V^{*} + \frac{c^{*2}}{2\beta_{1}^{*}} 
\end{align}
From \textit{theorem 4.18} in \cite{khalil2002nonlinear}, the solution of the error system $\Tilde{W}^{*}(t)$ is UUB. 
\end{proof}

\begin{remark} 
From \eqref{eq:bound_on_Wtilde_star}, $\tilde{W}^{*}(t) \in \mathcal{L}_{\infty}$ with an ultimate-bound $\frac{c^*}{\beta_{1}^{*}}$. Using the Projection operation, $\norm{\tilde{W}^{*}(t)}\leq 2\alpha^{*} + \epsilon^{*}$. Hence, the IE condition will dictate a smaller upper bound of $\norm{\tilde{W}^{*}(t)}$ when
\begin{align}
    \norm{\tilde{W}^{*}(t)} \leq \frac{c^{*}}{\gamma_{3}\lambda_{\min}(\Phi_{ff}(T))} \leq 2\alpha^{*} + \epsilon^{*}
    \label{eq:CIE_condition}
\end{align}
Therefore, using \eqref{eq:c_star} and \eqref{eq:CIE_condition} the following upper-bound on $\Bar{\delta}_{w}$ would have a guaranteed advantageous for IE-based design.  
\begin{align}
    \bar{\delta}_{W} \leq \frac{\gamma_{3}(2\alpha^{*} + \epsilon^{*})\lambda_{\min}(\Phi_{ff}(T))p_{f}}{\bar{\phi}^{2}}\big(\frac{1}{\frac{\gamma_{1}}{p_{f}} + \frac{\gamma_{2} + \gamma_{3}}{p_{ff}}} \big)
    \label{eq:delta_w_bound}
\end{align}
It can be inferred that sufficient degree of excitation $\gamma_{IE}$, which directly affects the magnitude of $\lambda_{\min}(\Phi_{ff}(T))$, ensures satisfaction of the above inequality. 
\end{remark}

Note that this also provides a smaller upper bound on $W_{r}(t)$ in Theorem \ref{theorem:1}.

\begin{remark}
Our proposed formulation of dual adaptation has structural similarity with recent development in higher-order adaptive control formulation \cite{gaudio2019provably, boffi2019higher}, where there are two update laws consisting of a surrogate variable and the actual parameter estimator. The surrogate variable updates according to the actual parameter estimator law from traditional MRAC while the actual parameter estimator chases this surrogate variable. In the proposed method, the total estimator $\hat{W}(t)$ is analogously chasing the nominal estimator $\hat{W}^*(t)$. However, a detailed analytical comparison between these two techniques is yet to be explored and will be considered in future research.
\end{remark}
\subsection{Performance Recovery in Disturbance-free Scenario $\delta_W(t)=0$}

Consider the following result as a special case having time-invariant system, i.e., $W(t)=W^*$ with no perturbation term $\delta_{w}(t)$. 
\begin{theorem}
Using the system dynamics in \eqref{eq:x_dot}, controller design in \eqref{eq:controller}, the unknown parameter $W(t) = W^{*}$, the update laws in \eqref{eq:W_hat_dot} and \eqref{eq:W_hat_star_dot} guarantees exponential convergence in the extended state space $[e^{T}(t), \tilde{W}^{T}(t)]^{T}$, provided the IE condition in Assumption \ref{ass:IE} is satisfied.
\end{theorem}

\begin{proof}
The derivative along system trajectories of \eqref{eq:V_candidate} with $W(t) = W^{*}$ gives

\begin{align}
    \dot{V} = &  e^{T}(t)(A_{m}^{T}P + PA_{m})e(t) - 2e^{T}(t)PB\Tilde{W}^{T}(t)\phi(x(t))  \nonumber
    \\
    & + 2Tr\Big(\Tilde{W}^{T}(t)\big(\phi(x(t))e^{T}(t)PB - \sigma(\hat{W}(t)-\hat{W}^{*}(t))\big) \nonumber
    \\
    = & -e^{T}(t)Qe(t) - 2\sigma Tr\big(\Tilde{W}^{T}(t)\tilde{W}(t)\big) \nonumber
    \\
    & + 2\sigma Tr\big(\Tilde{W}^{T}(t)\tilde{W}^{*}(t))\big) \nonumber
\end{align}
Using $2Tr\big(\Tilde{W}^{T}(t)\tilde{W}^{*}(t))\big) \leq \norm{\tilde{W}(t)}^{2}_{F} + \norm{\tilde{W}^{*}(t)}^{2}_{F}$
\begin{align}
    \dot{V} \leq & -\lambda_{\min}(Q)\norm{e(t)}^{2} - \sigma\norm{\tilde{W}(t)}_{F}^{2} + \sigma\norm{\tilde{W}^{*}(t)}_{F}^{2} \nonumber
    \\
    \leq & -\lambda_{\min}(Q)\norm{e(t)}^{2} - \sigma\norm{\tilde{W}(t)}_{F}^{2} + \underbrace{\sigma \lambda_{\max}(\Gamma_{W^{*}})}_{c_{W}^{*}}V^{*} \nonumber
    \\
    \leq & -\frac{\beta_{1}}{\beta_{2}}V + c_{W}^{*}V^{*}
\label{eq:V_dot_cororllary}
\end{align}
where, $\beta_{1}, \beta_{2}$ are same as in \eqref{eq:V_dot_UUB}, $c^{*}_{W} > 0$ and $V^{*}(t) $ is from \eqref{eq:V_star}.


Analysing the convergence rate of $V^{*}(t)$ when $W(t) = W^{*}$, we remove the time-varying term from \eqref{eq:lyap_star_dot} ($c^{*}$ becomes $0$), to get
\begin{align}
\begin{split}
        & \dot{V}^{*} \leq  -\underbrace{\frac{2\beta_{1}^{*}}{\beta_{2}^{*}}}_{c_{\Omega}^{*}}V^{*}  
    \\
    \implies & V^{*} \leq \exp(-c_{\Omega}^{*}t)V^{*}(T) \, \forall t\geq T
\end{split}
    \label{eq:V_star_dot_no_time_varying}
\end{align}


where, $c_{\Omega}^{*} \in \mathbb{R}_{>0}$.

Finally, integrating \eqref{eq:V_dot_cororllary} and replacing in \eqref{eq:V_star_dot_no_time_varying} we get, (showing explicit time dependence for clarity) 
\begin{align}
        V(t) \leq & \exp\{-\frac{\beta_{1}}{\beta_{2}}(t - T)\}V(T) + \nonumber
        \\
        & c_{W}^{*}\int_{T}^{t}\exp\{-\frac{\beta_{1}}{\beta_{2}}(t - \tau)- c_{\Omega}^{*}\tau\}d\tau V^{*}(T) 
        \\
        = & \exp\{-\frac{\beta_{1}}{\beta_{2}}(t - T)\}V(T) + \frac{c_{W}^{*}}{c_{\Omega}}\exp\{-c_{\Omega}^{*}t\}V^{*}(T) \nonumber
        \\
        & \underbrace{-c_{W}^{*}\exp\{-\frac{\beta_{1}}{\beta_{2}}t\}\exp\{c_{\Omega}T\}V^{*}(T)}_{\leq 0} \nonumber        
        \\
        \leq &  \exp\{-\frac{\beta_{1}}{\beta_{2}}(t - T)\}V(T) + \frac{c_{W}^{*}}{c_{\Omega}}\exp\{-c_{\Omega}^{*}t\}V^{*}(T) \nonumber
    \label{eq:V_corollary_final1}
\end{align}
where, $c_{\Omega} = \frac{\beta_{1}}{\beta_{2}}- c_{\Omega}^{*}$ and $V(t) \to 0 $ as $t \to \infty$.



\end{proof}

\begin{remark}
The above Theorem is a crucial feature of the proposed scheme entailing performance recovery in the disturbance free case. The result indicates that dual adaptation mechanism can ensure parameter convergence ($\tilde{W}(t)\to 0$ and $\tilde{W}^*(t)\to 0$) under the IE condition when $\delta_{W}(t)=0$ and thereby also invoke exponential rate of convergence for tracking error $e(t)$. The proposed algorithm successfully unifies the cases of time-varying and time-invariant parameter using the dual adaptation principle unlike the recent result \cite{gaudio2021parameter}.
\end{remark}

\section{Conclusion}
\label{sec:conclusion}
In this paper we design a composite adaptive controller using dual adaptation technique for time-varying dynamical systems. The novel dual adaptation mechanism can efficiently deal with the time-varying nature of the parameters. A combined approach using Projection and $\sigma$-modification algorithms is conceptualized while exploiting congelation of variables to claim a global result for the tracking error space. Unlike the classical adaptive systems requiring the restrictive PE condition for accurate parameter estimation, the proposed work builds on the milder IE condition. A rigorous Lyapunov stability analysis is performed to ensure UUB stability of the closed-loop system. Moreover, the proposed algorithm can recover the performance of previously designed IE-based adaptive controller in case of plants having constant parameters. 

\appendices




\ifCLASSOPTIONcaptionsoff
  \newpage
\fi



%
\bibliography{ref}

\begin{thebibliography}{10}

\bibitem{narendra2012stable}
K.~S. Narendra and A.~M. Annaswamy, {\em Stable adaptive systems}.
\newblock Courier Corporation, 2012.

\bibitem{ioannou1996robust}
P.~A. Ioannou and J.~Sun, {\em Robust adaptive control}, vol.~1.
\newblock PTR Prentice-Hall Upper Saddle River, NJ, 1996.

\bibitem{Dixon}
O.~S. Patil, R.~Sun, S.~Bhasin, and W.~E. Dixon, ``Adaptive control of
  time-varying parameter systems with asymptotic tracking,'' {\em IEEE
  Transactions on Automatic Control}, pp.~1--1, 2022.

\bibitem{lavretsky2013robust}
E.~Lavretsky and K.~A. Wise, ``Robust adaptive control,'' in {\em Robust and
  adaptive control}, pp.~317--353, Springer, 2013.

\bibitem{lavretsky2011projection}
E.~Lavretsky and T.~E. Gibson, ``Projection operator in adaptive systems,''
  {\em arXiv preprint arXiv:1112.4232}, 2011.

\bibitem{chen2020adaptive2}
K.~Chen and A.~Astolfi, ``Adaptive control for systems with time-varying
  parameters,'' {\em IEEE Transactions on Automatic Control}, 2020.

\bibitem{arabi2019set}
E.~Arabi and T.~Yucelen, ``Set-theoretic model reference adaptive control with
  time-varying performance bounds,'' {\em International Journal of Control},
  vol.~92, no.~11, pp.~2509--2520, 2019.

\bibitem{slotine1989composite}
J.-J.~E. Slotine and W.~Li, ``Composite adaptive control of robot
  manipulators,'' {\em Automatica}, vol.~25, no.~4, pp.~509--519, 1989.

\bibitem{lavretsky_composite}
E.~Lavretsky, ``Combined/composite model reference adaptive control,'' {\em
  IEEE Transactions on Automatic Control}, vol.~54, no.~11, pp.~2692--2697,
  2009.

\bibitem{chowdhary2013concurrent}
G.~Chowdhary, T.~Yucelen, M.~M{\"u}hlegg, and E.~N. Johnson, ``Concurrent
  learning adaptive control of linear systems with exponentially convergent
  bounds,'' {\em International Journal of Adaptive Control and Signal
  Processing}, vol.~27, no.~4, pp.~280--301, 2013.

\bibitem{kamalapurkar2017concurrent}
R.~Kamalapurkar, B.~Reish, G.~Chowdhary, and W.~E. Dixon, ``Concurrent learning
  for parameter estimation using dynamic state-derivative estimators,'' {\em
  IEEE Transactions on Automatic Control}, vol.~62, no.~7, pp.~3594--3601,
  2017.

\bibitem{parikh2019integral}
A.~Parikh, R.~Kamalapurkar, and W.~E. Dixon, ``Integral concurrent learning:
  Adaptive control with parameter convergence using finite excitation,'' {\em
  International Journal of Adaptive Control and Signal Processing}, vol.~33,
  no.~12, pp.~1775--1787, 2019.

\bibitem{ortega2020new}
R.~Ortega, S.~Aranovskiy, A.~A. Pyrkin, A.~Astolfi, and A.~A. Bobtsov, ``New
  results on parameter estimation via dynamic regressor extension and mixing:
  Continuous and discrete-time cases,'' {\em IEEE Transactions on Automatic
  Control}, vol.~66, no.~5, pp.~2265--2272, 2020.

\bibitem{roy2017combined}
S.~B. Roy, S.~Bhasin, and I.~N. Kar, ``Combined mrac for unknown mimo lti
  systems with parameter convergence,'' {\em IEEE Transactions on Automatic
  Control}, vol.~63, no.~1, pp.~283--290, 2017.

\bibitem{roy2016parameter}
S.~B. Roy, S.~Bhasin, and I.~N. Kar, ``Parameter convergence via a novel
  pi-like composite adaptive controller for uncertain euler-lagrange systems,''
  in {\em Conference on Decision and Control (CDC)}, pp.~1261--1266, IEEE,
  2016.

\bibitem{roy2017uges}
S.~B. Roy, S.~Bhasin, and I.~N. Kar, ``A uges switched mrac architecture using
  initial excitation,'' {\em IFAC World Congress}, vol.~50, no.~1,
  pp.~7044--7051, 2017.

\bibitem{roy2019robustness}
S.~B. Roy and S.~Bhasin, ``Robustness analysis of initial excitation based
  adaptive control,'' in {\em 2019 IEEE 58th Conference on Decision and Control
  (CDC)}, pp.~7055--7062, IEEE, 2019.

\bibitem{gaudio2021parameter}
J.~E. Gaudio, A.~M. Annaswamy, E.~Lavretsky, and M.~Bolender, ``Parameter
  estimation in adaptive control of time-varying systems under a range of
  excitation conditions,'' {\em IEEE Transactions on Automatic Control}, 2021.

\bibitem{chen2018adaptive}
K.~Chen and A.~Astolfi, ``Adaptive control of linear systems with time-varying
  parameters,'' in {\em 2018 Annual American Control Conference (ACC)},
  pp.~80--85, IEEE, 2018.

\bibitem{pan2018composite}
Y.~Pan and H.~Yu, ``Composite learning robot control with guaranteed parameter
  convergence,'' {\em Automatica}, vol.~89, pp.~398--406, 2018.

\bibitem{Yucelen_acc}
E.~Arabi and T.~Yucelen, ``A generalization to set-theoretic model reference
  adaptive control architecture for enforcing user-defined time-varying
  performance bounds,'' in {\em 2017 American Control Conference (ACC)},
  pp.~5077--5082, 2017.

\bibitem{khalil2002nonlinear}
H.~K. Khalil and J.~W. Grizzle, {\em Nonlinear systems}, vol.~3.
\newblock Prentice hall Upper Saddle River, NJ, 2002.

\bibitem{gaudio2019provably}
J.~E. Gaudio, T.~E. Gibson, A.~M. Annaswamy, and M.~A. Bolender, ``Provably
  correct learning algorithms in the presence of time-varying features using a
  variational perspective,'' {\em arXiv preprint arXiv:1903.04666}, 2019.

\bibitem{boffi2019higher}
N.~M. Boffi and J.-J.~E. Slotine, ``Higher-order algorithms for nonlinearly
  parameterized adaptive control,'' {\em arXiv preprint arXiv:1912.13154},
  2019.

\end{thebibliography}
\bibliographystyle{ieeetr}




%








\end{document}